\documentclass[aps,pra,reprint,superscriptaddress]{revtex4-2}
\usepackage{amssymb}
\usepackage{mathrsfs}
\usepackage{graphicx}
\usepackage{algorithm}
\usepackage[normalem]{ulem}
\usepackage{color}
\usepackage{bm}
\usepackage{physics}
\usepackage{amsfonts}
\usepackage{mathrsfs}
\usepackage{lipsum}
\usepackage[]{qcircuit}
\usepackage{algpseudocode}
\usepackage[english]{babel}
\usepackage{amsthm}
\newtheorem{theorem}{Theorem}
\newtheorem{lemma}[theorem]{Lemma}

\newcommand{\rv}[1]{\textcolor{black}{#1}}

\begin{document}

\preprint{APS/123-QED}

\title{A penalty-free quantum algorithm to find energy eigenstates}


\author{Nannan Ma}
\email{e0408695@u.nus.edu}
\affiliation{Centre for Quantum Technologies, National University of Singapore, Singapore 117543, Singapore}

\author{Heng Dai}
\affiliation{Department of Physics, National University of Singapore, Singapore 117551, Singapore}

\author{Jiangbin Gong}
\email{phygj@nus.edu.sg}
\affiliation{Department of Physics, National University of Singapore, Singapore 117551, Singapore}
\affiliation{Centre for Quantum Technologies, National University of Singapore, Singapore 117543, Singapore}

\date{\today}

\begin{abstract}

Finding eigenstates of a given many-body Hamiltonian is a long-standing challenge due to the perceived computational complexity.  Leveraging on the hardware of a quantum computer accommodating the exponential growth of the Hilbert space size with the number of qubits, more quantum algorithms to find the eigenstates of many-body Hamiltonians will be of wide interest with profound implications and applications.  In this work, we advocate a quantum algorithm to find the ground state and excited states of many-body systems,  without any penalty functions, variational steps or hybrid quantum-classical steps.  Our fully quantum algorithm will be an important addition to the quantum computational toolbox to tackle problems intractable on classical machines.

\end{abstract}


\maketitle
\section{Introduction}

The eigenvalues and eigenstates of a many-body Hamiltonian are fundamental elements in both physics research \cite{lin1993exact,qi2011topological,thogersen2004coupled,schuch2009computational} and computational applications \cite{mohseni2022ising,lucas2014ising}. In physics, eigenstate information provides essential insight into a system’s properties and is key to controlling the associated quantum processes. For example, in low-temperature physics \cite{verstraete2004matrix,narasimhachar2015low}, low-lying energy eigenstates dominate the physical behavior, whereas in quantum annealing \cite{hauke2020perspectives,boixo2014evidence}, the minimum energy gap between the ground state and the first excited state determines the adiabaticity of the annealing protocol. Beyond physics, many practical computational problems can be reformulated as finding the eigenvalues and eigenstates of a Hamiltonian. A notable example is the Quadratic Unconstrained Binary Optimization (QUBO) problem \cite{barahona1988application,glover2018tutorial}. 

For an $n$-qubit Hamiltonian system, solving the eigenstates classically requires $O(2^{2n})$ storage and $O(2^{3n})$ computational time in general, making it a real challenge for classical machines as $n$ grows \cite{cuppen1980divide,ceperley1980ground,orus2014practical}. In contrast, quantum algorithms require only $O(n)$ qubits for storage.  To date it is still unclear how to best leverage quantum computers \cite{preskill2023quantum} to solve all eigenstates systematically. 
One known algorithm is the Variational Quantum Eigensolver (VQE) \cite{kandala2017hardware,tilly2022variational, kirby2021variational}, a hybrid quantum-classical algorithm, as it utilizes a parameterized quantum circuit (ansatz) and classical optimization to approximate eigenstates.   The calculation of excited states needs extra considerations, and indeed in variants of VQE, e.g. Neural Quantum Digital Twins \cite{lu2025neural} and Variational Quantum Deflation \cite{higgott2019variational}, one needs to implement a penalty term representing some overlap with the ground state to be added to the loss function, thereby suppressing contributions from the ground state. {Another strategy to find both ground and excited energy eigenstates is to search for the subspace spanned by low-lying eigenstates \cite{nakanishi2019subspace} and then to solve the generalized eigenvalue problem within this subspace \cite{zhang2025unified,chen2025tangent}.}  All such efforts to date are motivating, but should not be regarded as having provided general solutions, because the so-called barren plateaus present a big challenge in all optimization-based approaches \cite{larocca2025barren,mcclean2018barren}, consistent with the inherent complexity of the eigenenergy-eigenstate problem itself.   Though not suffering from the barren plateaus, quantum annealing however requires a sufficiently long evolution time to ensure adiabaticity and some hybrid approaches can be invoked to estimate the gap between the ground state and the excited states.  
That is, even in quantum annealing where the ground state is the focus,  the energy gap between the ground state and the excited states along the entire adiabatic pathway is still a crucial element \cite{chancellor2024experience,albash2018adiabatic,karacan2025filter}. \rv{There are also other methodologies aiming at computing excited states, such as Quantum Filter Diagonalization (QFD) \cite{parrish2019quantum} and Filtered Quantum Phase Estimation (FQPE) \cite{sakuma2026quantum}. However, the efficacy of QFD is highly dependent on the selection and preparation of appropriate reference states, whereas FQPE requires the construction of an optimized filter function tailored to the system Hamiltonian.}

It is also necessary to mention two other algorithms aiming at the ground state: Imaginary Time Evolution (ITE) \cite{xie2024probabilistic, turro2022imaginary,mao2023measurement,gangat2017steady,lin2021real} and Quantum Imaginary Time Evolution (QITE) \cite{jones2019variational,chai2025optimizing,motta2020determining,sun2021quantum,kamakari2022digital}.
Both methods implement the non-unitary evolution operator $e^{-tH}$, where $t$ is the evolution time and $H$ is the system Hamiltonian, but in different ways. ITE implements $e^{-tH}$ directly, for example, by using embedding the targeted time evolution in a larger unitary operator. By contrast, QITE designs a specific unitary operator, without dilation but still requiring measurements on the quantum circuit, so as to produce the sought-after mapping from an initial state $|\psi_0\rangle$ to the final state $e^{-tH}|\psi_0\rangle/\Tr[e^{-tH}|\psi_0\rangle]$.  For both ITE and QITE,  as the imaginary evolution time progresses the system state gradually converges to the true ground state, without the need for careful path design or constraints on evolution speed. Yet such path-free convergence comes at a cost—ITE as described above can suffer from scalability for truly large-scale problems, because it can be experimentally challenging to realize the operator $e^{-tH}$ as part of a larger unitary operator.  But most important of all, ITE to date offers no straightforward extension to find the excited states.  QITE requires to frequently adjust, based on measurement outcomes from the current state,  quantum operations in real time,  either through variational optimization {\cite{nishi2021implementation, benedetti2021hardware}} 
or by solving partial differential equations to determine the circuit parameters \cite{mcardle2019variational}. 
To find excited states, QITE \cite{jones2019variational} typically relies on a penalty term, similar to VQE. This step requires classical post-processing and often demands a highly expressive ansatz or network structure.  The power of QITE is hence also yet to be improved.  \rv{We also note the Quantum Monte Carlo (QMC) method \cite{zhang2025quantum,huggins2022unbiasing,kanno2024quantum} as a well-established framework for implementing ITE stochastically. In this approach, a set of stochastic samples is evolved according to probabilities derived from the system Hamiltonian at each iteration. QMC faces significant challenges, most notably the sign problem that can hinder convergence.}

We propose in this work a quantum algorithm with an ambitious goal: to find both the ground state and excited eigenstates of a many-body Hamiltonian, without using any penalty function.   Since there are no variational steps involved, there will be no encounter with the barren plateaus.  Specifically, a wide class of many-body Hamiltonians $H$ can be modeled as a sum of two components: $H = h_1 + h_2$, with the eigenvalues and eigenstates of $h_1$ and $h_2$ solvable and physically easy to prepare on a quantum circuit,  e.g., the $H$ of the transverse Ising model as one obvious example.   We shall show that energy eigenstates of $H$ can then be found by a series of quantum operations on the quantum circuit.  To that end, our starting point is to shift and rescale the eigenvalues of $h_1$ and $h_2$ so that the ITE of $H$ can be implemented by stochastic sampling \cite{ye2025quantum} of a density matrix composed of projectors of $h_1$ and $h_2$.  The ground state can then be prepared on the quantum circuit from stochastic sampling based ITE.  To find excited states (in principle, more than one),  the main challenge is to project out the highly complex ground state or other lower excited states. Fortunately, this hurdle can now be overcome by making use of a latest protocol called the state-based simulation \cite{alipour2025state}.  Once ground-state projector is realized, the first excited state can be found on the quantum circuit in the same fashion, that is, with stochastic sampling based ITE again, without using any loss function, either.  This process can in principle continue so as to find more excited eigenstates of $H$.  Eigenvalues can be also ``read out"  from the quantum circuits since the expectation values of $h_1$ and $h_2$ can be directly measured.   As elaborated below, a combination of stochastic sampling algorithm based on the simple projectors of $h_1$ and $h_2$ and the state-based simulation to implement projectors of already-solved low-lying eigenstates enables us to find not just the ground state, but multiple excited eigenstates of a big class of many-body Hamiltonians.
\rv{Indeed, the assumed splitting of $H$ into two solvable parts represents a natural situation in many physics-oriened problems and does not limit the scope of relevance of our findings. First, our methodology can be extended to cases where $H$ consists of additional terms $h_i$. The eigenstates of $h_i$ should be with lower entanglement, therefore, the above-mentioned method such VQE or some subspace method can calculate such eigenstates with a circuit with reasonable depth. 
Second, as a fundamental result in computation theory, models within the QMA-complete complexity class can be expressed in this form \cite{PhysRevA.78.012352}.} 
\rv{For Hamiltonians that are not easily expressed in the aforementioned form, we propose a variant in which the Hamiltonian is treated as a linear combination of multi-qubit Pauli operators. In this approach, stochastic sampling is applied directly to these Pauli operators. This variant remains consistent with our primary algorithm and allows for straightforward implementation on a quantum circuit.} We finally use the transverse Ising model to elaborate the execution of our approach. 

\section{algorithm}

\subsection{Mapping $H$ to a mixed-state density matrix $\rho$}
We consider an $n$-qubit system with Hamiltonian $H=h_1+h_2$ in the Hilbert space ${\cal H}$.   Without loss of much generality, we assume that the spectrum decomposition of both $h_1$ and $h_2$ are obvious and easy to prepare on such a system.  
Let $\lambda^j_i$ and $\ket{\psi^j_i}$ denote the $i$-th eigenvalue and the corresponding eigenvector of $h_j$, where $j \in \{1, 2\}$.  In general, $[h_1,h_2] \neq 0$, and precisely this non-commutation relation prevents us from directly obtaining the eigenstates of $H$, even when we have full access to $\lambda^j_i$ and $\ket{\psi^j_i}$.  To exploit the availability of  $\lambda^j_i$ and $\ket{\psi^j_i}$,  we will not just regard $H$ as an abstract operator.  Rather, we shall view $H$ as an object equivalent to a mixed-state density matrix after some rescaling and shifting are applied to  $\lambda^j_i$.  In particular, a simple $h_j$ can be expressed as a linear combination of projectors $P^j_i = \ket{\psi^j_i}\bra{\psi^j_i}$, with the corresponding eigenvalues $\lambda^j_i$ serving as coefficients.  Because $H=h_1+h_2$,  $H$ is a sum of all such projectors, with coefficients $\lambda^j_i$.   To see that $H$ is actually equivalent to a legitimate mixed-state density matrix,  we may, without affecting the eigenstates of $H$,  shift all the $\lambda^j_i$ to make them non-negative so as to represent probabilities, and further rescale all of these coefficients to achieve probability normalization.  Specifically,    if we let $\bar{\lambda}_i^j = \lambda_i^j+ c^j$ by adding a constant offset $c^j$ to $h_j$ and then normalize the adjusted eigenvalues to define a probability distribution of the projectors, we have
\begin{equation}
    \rho = \sum_{i,j} p_i^j \ket{\psi_i^j}\bra{\psi_i^j} = \frac{1}{C} \left(H + (c^1 + c^2)I \right),
\end{equation}
where the ``probabilities'' $p_i^j = \bar{\lambda}_i^j / C$ and $C = \sum_{i,j} \bar{\lambda}_i^j$ is the normalization factor.  
Because we have aimed to convert the eigenvalues $\{\lambda_i^j\}$ into a valid probability distribution $\{p_i^j\}$, we require $c^j \ge -\min_i \lambda_i^j$ to ensure all $\bar{\lambda}_i^j = \lambda_i^j + c^j$ are non-negative.  This is a trivial step since we preassumed that all $\lambda^j_i$ are known or computable beforehand.   This way $\rho$ defined above is a legitimate mixed state density matrix by construction and its eigenstates are identical with those of $H$.   We have thus converted the problem to finding all eigenstates of $\rho$. 
To optimize the efficiency of the execution of our algorithm, we recommend to choose $c^j = -\min_i \lambda_i^j$. 

\subsection{ITE with stochastic sampling from a mixed-state density matrix}

The next step in our algorithm is to implement $e^{-\eta\rho}$ through ITE, with time step $\eta << 1$ via sampling. At each step, according to the probability distribution \( p_i^j \), we randomly sample a projector \( P_{i}^{j}=\ket{\psi_i^j}\bra{\psi_i^j} \) from the mixed state $\rho$. From one individual $P_{i}^{j}$ sampled, we implement the evolution operator $I-\eta P_{i}^{j}$.  Repeating this many steps as the imaginary time evolution proceeds, the ensemble average of the resulting time-evolution operator is $E\left[I-\eta P_{i}^{j}\right]=I-\eta \sum_{i,j} p_i^j P_{i}^j = I-\eta \rho \approx e^{-\eta \rho}$.

After outlining the stochastic sampling approach first considered in \cite{ye2025quantum}, we  remark on some detailed aspects in the execution of  the evolution operator  $I-\eta P_{i}^{j}$ for each sampled projector $P_{i}^{j}$.  This is a rather standard step and  it can be done by a post selection of a designed unitary operator $U$ on a slightly larger system $\mathcal{H}\otimes\mathcal{H}^a$ where $\mathcal{H}^a$ represents a one-qubit ancillary system.  The said $U$ can be the following:

\begin{equation}
\label{qcu}
    U=(I-P)\otimes I+P\otimes \left( (1 - \eta)I + i\sqrt{2\eta - \eta^2}\,\sigma_y \right).
\end{equation}

\begin{figure}
    \centering
     \scalebox{1.3}{
    \Qcircuit @C=1em @R=.7em {
    \lstick{}&\qw&\multigate{3}{U_{\text{tran}}^\dag}&\qw &\ctrl{1}&\qw & \multigate{3}{U_{\text{tran}}}&\qw &\qw \\ \lstick{}&\qw&\ghost{U_{\text{tran}}^\dag}&\qw &\ctrl{2}&\qw & \ghost{U_{\text{tran}}}&\qw &\qw \\ \lstick{}&\cdots & \nghost{U_{\text{tran}}^\dag}&\cdots&\cdots&\cdots & \nghost{U_{\text{tran}}}&\cdots& \quad
    \inputgroupv{1}{4}{.8em}{2.5em}{\rho_{0} }\\
   \lstick{}&\qw & \ghost{U_{\text{tran}}^\dag} &\qw&\ctrl{1}&\qw & \ghost{U_{\text{tran}}}&\qw &\qw
   \\
   \lstick{\ket{0}}&\qw &\qw&\qw &\gate{R_y}&\qw & \qw& \meter & \quad &\rstick{\ket{0}}
   }
    }
    \caption{\rv{Schematic of the $I-\eta P$ update. This circuit (defined by $U$ in Eq.~(\ref{qcu})) encodes a single step of the stochastic evolution. The target system, initially in state $\rho_{0}$, is coupled to a single-qubit ancilla initialized in $\ket{0}$. The transformation $U_{\text{tran}}$ maps the sampled operator $P$ to a canonical projector—specifically $\ket{1}\bra{1}^{\otimes n}$ as shown here. Upon completion of the circuit, the ancilla is measured in the computational basis; we post-select for the $\ket{0}$ outcome to herald the successful application of the $I-\eta P$ operation.}}
    \label{fig:cirillus} 
\end{figure}

The qualitatively meaning of $U$ designed above is to execute a $SU(2)$ operation conditioned upon that the $n$-qubit system is in the $n$-qubit pure state $P_i^j$.  Noting that the states $\ket{\psi_i^j}\bra{\psi_i^j}$ are assumed to be simple states of no or low entanglement content (e.g, all qubits are in state $|1\rangle$),
\rv{This implies that the unitary operator required to prepare the state $\ket{\psi_i^j}\bra{\psi_i^j}$ from the computational basis state $\ket{0}\bra{0}$ can be, in principle, relatively straightforward to implement on a quantum circuit.} $U$ is a $n$-qubit-controlled SU(2) gate, where the $n$-qubit system controls the single-qubit ancillary system, with the controlled operation being $(1 - \eta)I + i\sqrt{2\eta - \eta^2}\,\sigma_y $.
That is, if and only if the $n$-qubit system is indeed in the pure state $P_{i}^{j}$, apply the single-qubit unitary operator $
(1 - \eta)I + i\sqrt{2\eta - \eta^2}\,\sigma_y
$. The involved operation  $(1 - \eta)I + i\sqrt{2\eta - \eta^2}\,\sigma_y$ is a simple $R_y$ gate $e^{-i\frac{\theta}{2}\sigma_y}$ with the rotation angle $\theta=-2\text{arctan}(\frac{\sqrt{2\eta-\eta^2}}{1-\eta})$. This $n$-qubit controlled SU(2) unitary operator can be compiled with $O(n)$ CNOT gates \cite{vale2023decomposition}, feasible on a generic quantum computer with only $O(n)$ two-qubit gates. \rv{Though the algorithm exhibits linear scaling in terms of two-qubit gate counts from a theoretical perspective, its practical implementation remains constrained by the limited connectivity and the ancillary qubit requirements of current quantum hardware}.  The complete flow of our ITE approach with stochastic sampling is shown in Alg.~\ref{alg:the_alg}.

Let us now assume that the initial state of the ancillary qubit is fixed at $\ket{0}$ and there is a measurement on the ancillary qubit after applying $U$. By only keeping the  post selected ancillary-qubit state $\ket{0}$, the resulting effect on the $n$-qubit system is $\bra{0}U\ket{0}=I-\eta P$, the operation we need for the imaginary time evolution. The success rate  of this step is $\mathrm{Tr}((I-\eta P_{i}^{j})\rho_{\rm input}(I-\eta P_{i}^{j}))$, $\rho_{\rm input}$ represents a normalized input of the $n$-qubit system. 

During the execution of the above-described stochastic sampling algorithm,  different $P_{i}^{j}$ will be sampled and therefore one must execute different $n$-qubit controlled unitary operators $U$.  It is worth highlighting that this presents no operational challenge even for large-scale computational problems.  As assumed earlier, $P_i^{j}$ are simple pure-states projectors from eigenstates of $h_1$ and $h_2$ and conversion between these projectors can be very straightforward.  That is,  once a quantum circuit for one $n$-qubit controlled unitary operator is constructed, the other projectors sampled can be constructed with a simple unitary operator $U_{\text{tran}}$ as well, likely by adding more single-qubit operations to each of the $n$ qubits.  As one useful example connected with our simulation results below, the eigenstates of $h_j$ are basic computational basic states (no entanglement), such as $\ket{00011\ldots000}$ or $\ket{+++--\ldots+++}$.
In particular, if the controlled unitary gate for $\ket{000\ldots}\bra{000\ldots}$ has been implemented on a quantum circuit,  a different projector $\ket{+++\ldots}\bra{+++\ldots}$  can be obtained by applying the single qubit Hadamard gate to all the $n$ qubits of our system.  Each step of the ITE executed this way hence only requires a shallow circuit. 
\rv{A schematic representation of the quantum circuit for a single iteration is provided in Fig.~\ref{fig:cirillus}.} 

\begin{algorithm}[H]
    \caption{Imaginary Time Evolution by sampling a mixed-state density matrix}  
    \label{alg:the_alg}
    \begin{algorithmic}[1]
        \Require  
        Eigen-decompositions of $h_1$ and $h_2$: $\{\lambda_i^j, P_i^j = \ket{\psi_i^j}\bra{\psi_i^j}\}$ for $j \in \{1,2\}$;  
        Initial state: $\rho'_0$;  
        Time step: $\eta$;  
        Number of steps: $N$.
        \Ensure  
        Approximate ground state of $H = h_1 + h_2$
        
        \State Shift eigenvalues: $\bar{\lambda}_i^j = \lambda_i^j - \min_{i'} \lambda_{i'}^j$
        \State Normalize to obtain probability distribution:  
        $\{p_i^j = \bar{\lambda}_i^j / C, \; P_i^j\}$, where $C = \sum_{i,j} \bar{\lambda}_i^j$
        \State Initialize: $\rho_{\text{ITE}}' \gets \rho'_0$
        \For{$i = 1$ to $N$}
            \State Sample a projector $P$ from the distribution $\{p_i^j\}$
            \State Construct the corresponding quantum circuit $U_P$ as defined in Eq.~\ref{qcu}
            \State Apply $U_P$ to the joint state $\rho_{\text{ITE}}' \otimes \ket{0}\bra{0}$
            \State Measure the ancillary qubit
            \If{the measurement outcome is $\ket{0}\bra{0}$}
                \State Update state:  
                $\rho_{\text{ITE}}' \gets \dfrac{(I - \eta P)\rho_{\text{ITE}}'(I - \eta P)}{\operatorname{Tr}[(I - \eta P)\rho_{\text{ITE}}'(I - \eta P)]}$
            \Else
                \State Discard the state and \textbf{break}
            \EndIf
        \EndFor
        \State \Return Final state $\rho_{\text{ITE}}'$ as the approximated ground state
    \end{algorithmic}
\end{algorithm}


\rv{$\rho_0$ represents the initial state of the $n$-qubit system. Selecting an appropriate $\rho_0$ that exhibits a significant overlap with the target eigenstate will accelerate convergence; several established methods \cite{lin2020near, zhang2025quantumre} exist for this state preparation process.}
After $N$ steps of the above-described imaginary time evolution for a total evolution time $T = N\eta$,  one can estimate the total error arising from discretizing  the time.  Appendix A shows that the error is mainly determined by $\gamma = N\eta^2$.  One may examine the overall success rate of the whole algorithm, that is,  success for all the $N$ steps when performing post selection of the ancillary qubit state.   
\begin{equation}
    \mathrm{Tr}\left( \sigma_T(\rho_0) \right) + O(\gamma),
\end{equation}
where  $\sigma_T(\rho_0) = e^{-T\rho} \, \rho_0 \, e^{-T\rho}
$.

\subsection{Implementation of projectors of low-lying eigenstates}

 ITE with stochastic sampling as an innovative element, as depicted above, is already a useful algorithm to find the ground state (denoted $|E_g\rangle$)  and complications emerge when trying to find the excited states.  Conceptually, one intuitive strategy is to ``lift" the ground state of the Hamiltonian $H$ by adding a projector $P_g=|E_g\rangle\langle E_g|$ to the system Hamiltonian $H$ with sufficiently large coefficient (larger than the gap between the ground state and the first excited state). This way, finding the new ground state of the modified Hamiltonian by ITE will be equivalent to finding  the first excited state of $H$.  This thought line can continue for finding other excited state one by one.  For example,  in the context of our approach above, one may wish to include $P_g$ in the set of projectors $\{P_i^j\}$ and then apply operations like $I-\eta P_g$ whenever $P_g$ is sampled.
However, this is not feasible in practice because the ground state or other states already found on the quantum circuit is typically of high entanglement content and hence one cannot directly implement projectors like $P_g$.  In  variational algorithms aiming at finding the first excited state, one considers a penalty function to avoid the overlap with the ground state -- a procedure that can be computationally expensive for large scale problems with limited effectiveness, again due to the existence of barren plateaus.

We will now count on a recent key protocol \cite{alipour2025state}, the so-called state-based simulation, to realize operations $I-\eta P_g$ so as to cleanly project out the ground-state contributions during ITE by stochastic sampling.  Specifically, given that we can safely obtain the state $P_g$ on a quantum circuit by use of the new version of ITE here, we can next use the state-based simulation method to indeed implement the evolution $e^{-\eta P_g}$, at the price of introducing a dual $n-$ qubit system under identical setting.  That is, the whole setup involves three subsystems: the $n$-qubit original system $\mathcal{H}$, an $n$-qubit ancillary system $\mathcal{H}^a$, and a single-qubit control system $\mathcal{H}^c$.  We elaborate below that this suffices to realize the key imaginary time evolution step based on $P_g$: namely, $
\rho_1 = (I-\eta P_g) \rho_0 (I-\eta P_g)
$ (up to normalization) in the Hilbert space $\mathcal{H}$ of the system $H$. 

\rv{The circuit architecture for this state-based operation is depicted in Fig.~\ref{fig:stateillus}.}
Initially, the whole set up in the Hilbert space $\mathcal{H}\otimes\mathcal{H}^a\otimes\mathcal{H}^c$ is as follows: the original system is in the state $\rho_0$, the ancillary system is in $\rho_a = P_g$ obtained from our ITE, and the control system is in state $
\rho_c = (\ket{0} - \eta \ket{1})(\bra{0} - \eta \bra{1})
$.  The overall state is hence the direct product state $\rho_0\otimes P_g\otimes \rho_c$.

\begin{figure}
    \centering
     \scalebox{1.5}{
    \Qcircuit @C=1em @R=.7em {
    \lstick{}&\qw&\multigate{7}{U_\text{sw}} &\qw &\qw \\
    \lstick{}&\cdots & \nghost{U_\text{sw}}&\cdots& \quad
    \inputgroupv{1}{4}{.8em}{2.5em}{\rho_{\text{0}} }\\
   \lstick{}&\qw & \ghost{U_\text{sw}}&\qw &\qw
   \\
   \lstick{}&\qw&\ghost{U_\text{sw}}&\qw &\qw \\
    \lstick{}&\qw&\ghost{U_\text{sw}}&\qw &\qw \\
    \lstick{}&\cdots & \nghost{U_\text{sw}}&\cdots& \quad
    \inputgroupv{5}{8}{.8em}{2.5em}{P_g}\\
   \lstick{}&\qw & \ghost{U_\text{sw}}&\qw & \qw & \rstick{\text{Discard}}
   \\
   \lstick{}&\qw & \ghost{U_\text{sw}}&\qw &\qw
   \gategroup{5}{5}{8}{5}{.8em}{\}}
   \\
   \lstick{\ket{0}-\eta\ket{1}}&\qw &\ctrl{-1}& \meter & \rstick{\ket{+}}
   }
    }
    \caption{\rv{Schematic of the state-based simulation. The setup comprises three subsystems, $\mathcal{H}\otimes\mathcal{H}^a\otimes\mathcal{H}^c$: the primary system in state $\rho_{0}$, an ancillary system prepared in $\rho_a = P_g$, and a control qubit in the state $\rho_c \propto (\ket{0} - \eta \ket{1})(\bra{0} - \eta \bra{1})$. A controlled-SWAP gate is applied to the primary and ancillary subsystems, conditioned on the control qubit. Following the gate operation, the ancillary system is discarded, and the control qubit is measured. The procedure is successful upon obtaining the $\ket{+}$ outcome, which heralds the desired state evolution.}}
    \label{fig:stateillus} 
\end{figure}

The next step is to apply a controlled-swap gate: conditional upon the state of the control qubit, we swap the states between the original system and the ancillary system.  After this operation, we discard the ancillary system, leading to the following state of the system in $\mathcal{H} \otimes\mathcal{H}^c$ : $\mathrm{Tr}_{a}(U_{\text{csw}}\rho_0\otimes \rho_a\otimes \rho_c U^\dagger_{\text{csw}} )$.  The last step is to measure the control qubit in the $\ket{+}$ basis. The resulting post-measurement state (conditioned on the $\ket{+}$ outcome) is found to be:
\begin{equation}
\label{sbs}
    \frac{1}{2}(\rho_0 - \eta \{ \rho_0, P_g \}+\eta^2P_g) =\frac{1}{2}(I-\eta P_g)\rho_0(I-\eta P_g) +O(\eta^2),
\end{equation}
up to a normalized factor of $\frac{1}{2}\mathrm{Tr}(\rho_0 - \eta \{ \rho_0, P_g \}+\eta^2P_g)$.

The $1/2$ prefactor in Eq.~\ref{sbs} arises from the fact that the probability of finding the state $\rho_c$ being in the state $\ket{+}$ is always 50\% in each implementation step. \rv{That is, though the evolution $I-\eta P_{i}^{j}$ can be implemented using the state-based simulation, the factor of $1/2$ suggests that this approach should be avoided when possible in favor of the method in Algorithm \ref{alg:the_alg}. This justifies our choice to implement operators such as $I-\eta P_{i}^{j}$ via the previously described stochastic sampling-based ITE, reserving the state-based method specifically for $I-\eta P_g$ where it is necessary. }

\subsection{A penalty-free  algorithm for finding eigenstates above the ground state}
We are now ready to find the excited states of a many-body system represented by $H$.   Without loss of generality, we assume that our target is the first excited state.   Even finding the first excited state, together with the ground state already found, can already offer rich physical insights into a many-body system such as quantum phase transitions and adiabaticity for quantum annealing algorithms.  

Assuming that the ground state $P_g$ has been found and let us now return to the density matrix $\rho$ defined in Eq.~(1), the probability $p_{e1}$ associated with the first excited state projector $P_{e1}$ should be less than $\frac{1}{2^n - 1}$, as $(2^n-1)p_{e1} \leq 1-p_g\leq 1$.  In order to find the first excited,  we must suppress the contribution of the ground state projector $P_g$ so that the first excited state can emerge from the ITE.  To that end, we introduce a lifting procedure by assigning the ground-state sampling a probability $p_{\text{lift}} \geq \frac{1}{2^n - 1}$.  This modification can be done by adding to $\rho$ a new term associated with the ground state while rescaling all other probabilities:

\begin{algorithm}[H]
    \caption{Algorithm for finding the first excited eigenstate of a many-body $H$}  
    \label{alg:sb}
    \begin{algorithmic}[1]
        \Require  
        Eigenvalue distribution: $\{p_i^j, P_i^j = \ket{\psi_i^j}\bra{\psi_i^j}\}$ for $j \in \{1, 2\}$;  
        Ground state projector: $P_g$;  
        Initial state: $\rho'_0$;  
        Time step: $\eta$;  
        Number of steps: $N$;  
        Lift probability for $P_g$: $p_{\text{lift}} \ge \frac{1}{2^n - 1}$
        \Ensure  
        Approximate excited state of $H$

        \State Modify to the eigenvalue distribution $\{p'^j_i, p'_g\}$ in Eq.~ \ref{mp}
        \State Initialize: $\rho_{\text{ITE}}' \gets \rho'_0$
        \For{$i = 1$ to $N$}
            \State Sample a projector $P$ from the modified distribution 
            \If{$P \gets P_g$}
                
                \State Prepare the joint state $\rho_{\text{ITE}}' \otimes P_g \otimes \rho_c$
                \State Apply a controlled-swap gate $U_{\text{CSW}}$ between the system and ancilla, controlled by the control qubit; discard the ancillary system
                \State Measure the control qubit
                \If{the outcome is $\ket{+}\bra{+}$}
                    \State Update state:  
                    \[
                    \rho_{\text{ITE}}' \gets \frac{\rho_{\text{ITE}}' - \eta \{ \rho_{\text{ITE}}', P_g \} + \eta^2 P_g}{\operatorname{Tr}[\rho_{\text{ITE}}' - \eta \{ \rho_{\text{ITE}}', P_g \} + \eta^2 P_g]}
                    \]
                \Else
                    \State Discard the state and \textbf{break}
                \EndIf
            \Else
                \State Construct the corresponding quantum circuit $U_P$ as defined in Eq.~\ref{qcu}
                \State Apply $U_P$ to the joint state $\rho_{\text{ITE}}' \otimes \ket{0}\bra{0}$
                \State Measure the ancillary qubit
                \If{the outcome is $\ket{0}\bra{0}$}
                    \State Update state:  
                    \[
                    \rho_{\text{ITE}}' \gets \frac{(I - \eta P)\rho_{\text{ITE}}'(I - \eta P)}{\operatorname{Tr}[(I - \eta P)\rho_{\text{ITE}}'(I - \eta P)]}
                    \]
                \Else
                    \State Discard the state and \textbf{break}
                \EndIf
            \EndIf
        \EndFor
        \State \Return Final state $\rho_{\text{ITE}}'$ as the approximated excited state
    \end{algorithmic}
\end{algorithm}

\begin{equation}
    p'^j_i = \frac{p^j_i}{1 + p_{\text{lift}}}, \quad p'_g = \frac{p_{\text{lift}}}{1 + p_{\text{lift}}}.
    \label{mp}
\end{equation}
with  the mixed state density matrix now modified to the following: 

\begin{equation}
    \rho' = \sum_{i,j} p'^j_i \ket{\psi^j_i}\bra{\psi^j_i} + p'_g P_g.
\end{equation}
Certainly, $\rho'$ still contains precisely the same eigenstate information as the original many-body system $H$, however the ground state of $\rho'$, by construction, is actually the first excited state of $H$.   We next apply our stochastic sampling algorithm to $\rho'$ to implement ITE to find the ground state of $\rho'$, with the sampling probabilities given by  $p'^j_i$.   Importantly,  when the projector $P_g$ is sampled, we use the above-outlined state-based simulation method to implement the imaginary time evolution operator $e^{-\eta P_g}$.

The success rate of this procedure can be easily estimated and it is
\begin{equation}
    \frac{\mathrm{Tr}\left(\sigma_T(\rho_0)\right)}{2^{p'_g N}}+O(\gamma),
\end{equation}
where $\sigma_T(\rho_0)$ is the final state under the whole imaginary time evolution protocol, and $N$ denotes number of the steps (see Appendix B). From this success rate, one also sees a somewhat expected trade-off between the added weightage of the ground state to filter it out and the overall efficiency in accessing the first excited state.  In particular, though the number of steps $N$ seems to exponentially decrease the success rate, this concern can be lifted if we choose the smallest possible $p_g'\approx p_{\rm lift} \sim 2^{-n}$.  That is, the more qubits we have, the total number of steps we need to implement $e^{-\eta P_g}$ can still be small, and hence the prefactor entering the overall success rate, namely, $2^{-p'_g N}$ in the above equation will stay significant. \rv{The preceding analysis assumes an idealized, fault-tolerant environment. On current Noisy Intermediate-Scale Quantum (NISQ) platforms, performance is expected to degrade as system complexity increases. This deterioration is primarily driven by the accumulation of gate errors and decoherence effects throughout the iterative ITE process.}

\rv{There is also a slightly different approach as compared with that detailed in the Alg.~\ref{alg:sb}. We follow a fixed pattern when applying the operation $I-\eta P_g$ instead of using stochastic sampling.  That is,  on the one hand we execute  the ITE based on stochastic sampling of the density matrix $\rho$ defined in Eq.~(1); on the other hand, at certain predetermined steps, we implement the operator $I-\eta P_g$.  These operations serve the same purpose of filtering out the ground state of $H$. The overall proportion of these steps must exceed a threshold probability determined by $p_{\text{lift}}$.  As a simple strategy, one may apply $I-\eta P_g$ at every even-numbered step during the ITE. That is, the probability assigned to sampling $P_g$ is 50\% (more than the threshold), and this is now done deterministically.  This is fine in the proof-of-principle demonstrations (but not recommended in actual implementations on a quantum hardware), because filtering out the ground state too aggressively in actual experiments will lead to a too low success rate when implementing too many state-based simulation steps.}


To conclude this subsection, we briefly discuss how to extend the above considerations specifically for the first excited state to  other higher excited states.  One straightforward scenario is to require more copies of the original system $H$ and we then execute the above algorithm to prepare the ground state, the first eigenstate\rv{, etc., }on each of the copies.  One can then continue to implement projectors of those already excited states, combined with the stochastic sampling based ITE, to find the next excited state. 

\rv{Finally, we discuss techniques for stabilizing the obtained target state. The accuracy of our algorithm is intrinsically linked to the energy gap. In cases of near-degeneracy—particularly when the number of iterations is limited—the resulting state may manifest as a superposition within the near-degenerate manifold, retaining a non-negligible overlap with other states. To mitigate this, symmetry-sector projection \cite{stetcu2023projection} (if applicable) or subspace methods—such as Quantum Subspace Expansion \cite{mcclean2020decoding,takeshita2020increasing} or Lanczos-type subspace built-up \cite{kirby2023exact}—can be employed to isolate the target state and stabilize the final result}

\subsection{A variant to complement the main algorithm }
\rv{Here we elaborate on another way to implement the key elements of our approach. In general, a Hamiltonian can be decomposed as the linear combination of n-qubit Pauli operators $O_i=\pm \sigma_{i_1}\otimes\dots\otimes\sigma_{i_n}$ with the positive coefficients $\alpha_i$, that is $H=\sum_{i=1}^n\alpha_iO_i$. $O_i$ has only two different eigenvalues $\pm1$, and it can be transformed to be a special one-qubit operator $I^{n-1}\otimes \sigma_z$ using a unitary operator $U_{O_i}$. For example, with two-qubit gate $\ket{0}\bra{0}\otimes I+\ket{1}\bra{1}\otimes\sigma_z$, the Pauli operator $\sigma_z\otimes\sigma_z$ can be transform into $I\otimes\sigma_z$.
In general, for a $n$-qubit system, $U_{O_i}$ consists of $n-1$ two-qubit gates.
We can construct a distribution $p_i= \frac{\alpha_i}{C_\alpha}$ with $C_\alpha=\sum\alpha_i$.}

\rv{To encode ITE related operator into a unitary form, $H'=\sum_{i=1}^np_i(I+O_i)$ is used with the eigenstates unchanged, but lifting the eigenvalues uniformly by $I$. We still evolve $e^{-\eta H'}$ with a stochastic sampling. That is $E[I-\eta(I+O_j)]=I-\eta H' \approx e^{-\eta H'}$. The implementation of $I-\eta(I+O_j)$ is still standard and practical. It is also encoded into a designed unitary operator $U$ on a $n+1$-qubit system $\mathcal{H}\otimes\mathcal{H}^a$ where $\mathcal{H}^a$ is one-qubit ancillary system. $U$ is given as follows:
\begin{widetext}
    \begin{align}
    U&=U_{O_j}^\dagger \left(
    (I^{n-1}\otimes\ket{1}\bra{1})\otimes I+(I^{n-1}\otimes\ket{0}\bra{0})\otimes \left( (1 - 2\eta)I + 2i\sqrt{\eta - \eta^2}\,\sigma_y \right)\right)U_{O_j}\notag\\
    &=U_{O_j}^\dagger \left(
    (I^{n-1}\otimes\ket{1}\bra{1})\otimes I+(I^{n-1}\otimes\ket{0}\bra{0})\otimes R_y(\theta)\right)U_{O_j}
    \end{align}
\end{widetext}
with $\theta=-2\text{arctan}(\frac{\sqrt{\eta-\eta^2}}{1-2\eta})$.
The initial state of the ancillary system is $\ket{0}$, and after $U$, the ancillary system is post selected on $\ket{0}$, the final effect is $\bra{0}U\ket{0}=I-\eta(I+O_i)$.
The consequent steps are the same as explained in Algorithm~\ref{alg:the_alg}. To find the excited sate, we still use the method in Algorithm~\ref{alg:sb}.
}

\begin{figure}[htbp]
	\includegraphics[scale=0.07]{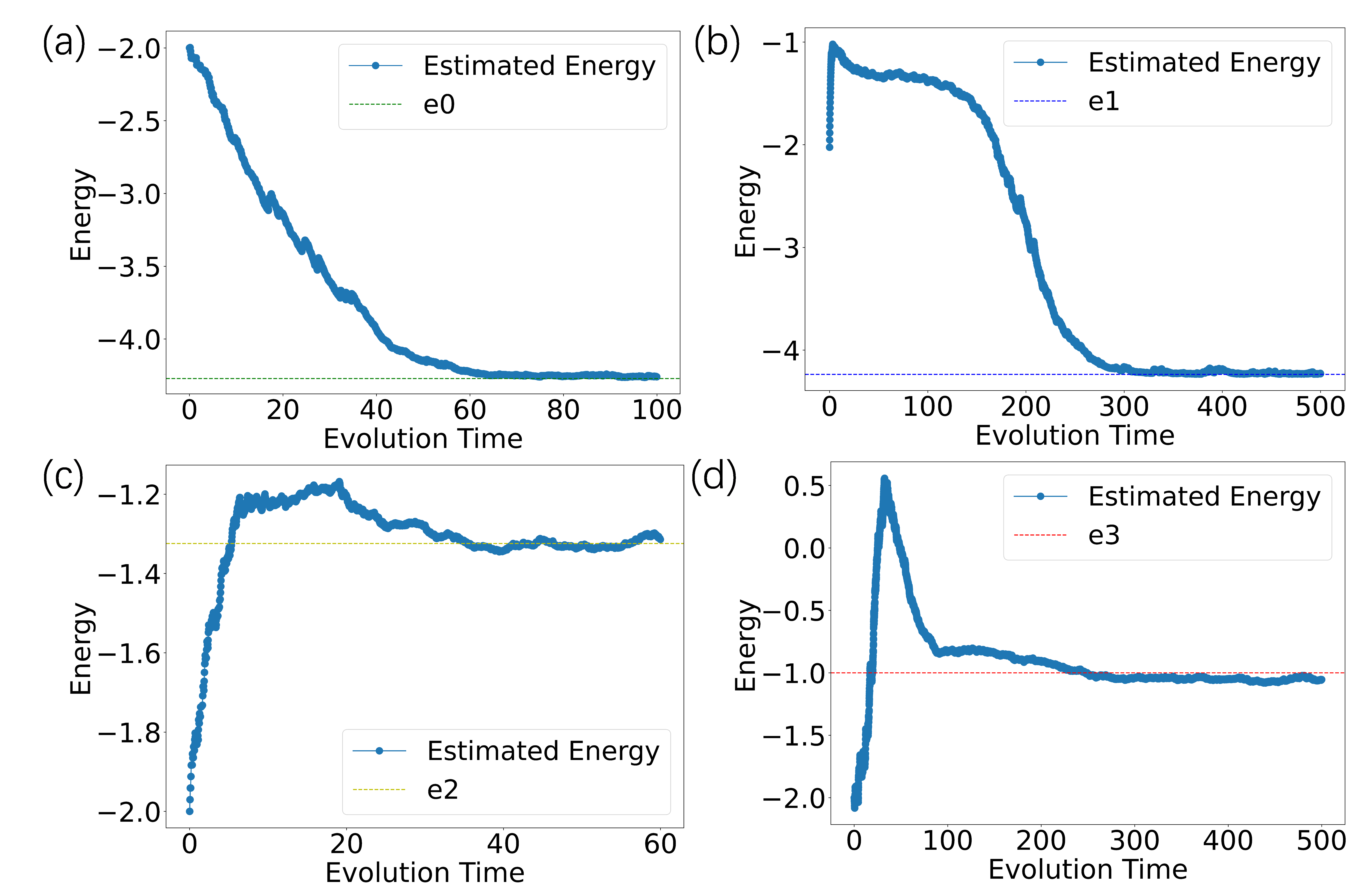}
	\caption{\label{4q} Evolution of the first 4 lower eigenstates of the 4-qubit Ising model with $J=1, B=0.5$. (a) The result for ground state with algorithm 1 as all projectors are easy to implemented. The fidelity is $99.9\%$ (b) The result for first excited state with algorithm 2 with the calculated ground state. The fidelity is $98.6\%$ (c) The result for second excited state with algorithm 2 with the calculated ground state and first excited state. The fidelity is $96.8\%$ (d) The result for third excited state with algorithm 2 with the calculated ground state and first and second excited state. The fidelity is $96.6\%$}
\end{figure}


\begin{figure*}[htbp]
	\includegraphics[scale=0.12]{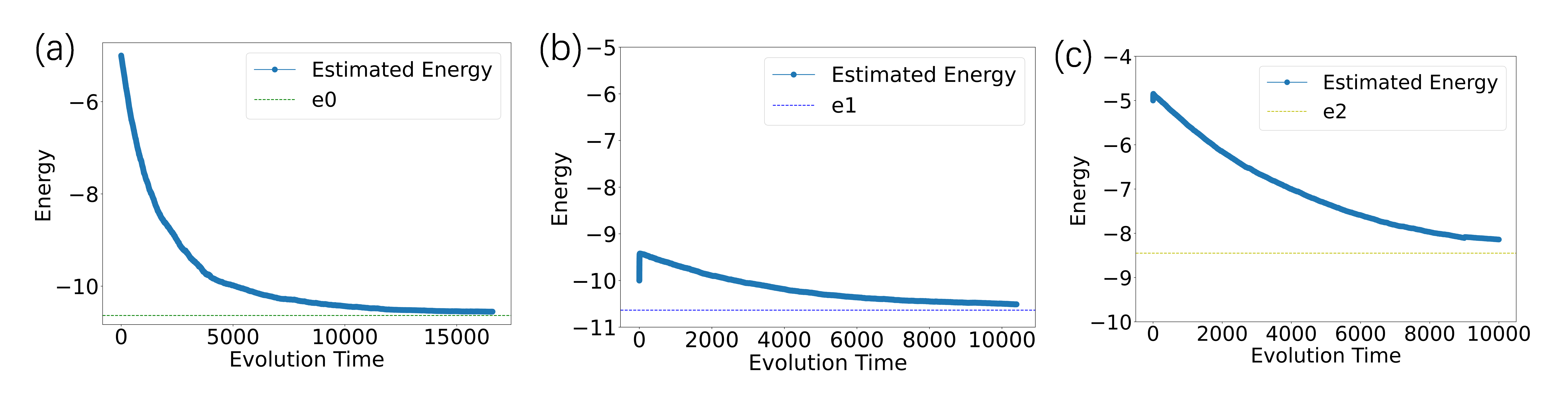}
	\caption{\label{q10} Evolution of the first 3 lower eigenstates of the 10-qubit Ising model with $J=1, B=0.5$. (a) The result for ground state with algorithm 1 as all projectors are easy to implemented. The fidelity is $98.6\%$ (b) The result for first excited state with algorithm 2 with the calculated ground state. The fidelity is $97.8\%$ (c) The result for second excited state with algorithm 2 with the calculated ground state and first excited state. The fidelity is $90.2\%$.}
\end{figure*}

\rv{At the end of this algorithm section, we give the resource estimation (see Appendix C) as an indicator for the practical implementation. The initial fidelity is $f_0$ and the target fidelity is $f_T$. for constructed $\rho$ of the target Hamiltonian $H$, the ground energy value is $\lambda_0$ and the gap is $\Delta\lambda$. To approach the target fidelity $f_T$ with a deviation $\Delta f$, we should set the hyper parameters at least:
\begin{equation}
    \begin{split}
        \eta &=\frac{\Delta\lambda\Delta f}{\ln{\left(\frac{f_T\sqrt{1-f_0^2}}{f_0\sqrt{1-f_T^2}}\right)}}\\
        N &=\frac{1}{\Delta f\Delta^2\lambda}\ln^2{\left(\frac{f_T\sqrt{1-f_0^2}}{f_0\sqrt{1-f_T^2}}\right)}.
    \end{split}
\end{equation}
At each step, there are O(n) two-qubit gate are applied, and 1 (n+1) ancillary qubit is needed for stochastic sampling ITE (state-based ITE). The success probability is up bounded by:
\begin{equation}
    \left(\frac{f_T\sqrt{1-f_0^2}}{f_0\sqrt{1-f_T^2}}\right)^{-2\frac{\lambda_0}{\Delta\lambda}}+O(\Delta f).
\end{equation}
}

\rv{
It is also worth noting that while the total number of time steps $N$ remains on the same order as existing ITE methods, the gate overhead per step differs significantly. General ITE methods typically rely on Trotter-Suzuki expansions, which require $O(n)$ substeps of $e^{\eta h_i}$ to simulate a single step $e^{\eta H}$. Since each substep requires $O(n)$ gates, the total complexity per step scales as $O(n^2)$. In contrast, our method maintains $O(n)$ scaling per step.
}

\section{Proof-of-principle demonstrations}
We now use the one-dimensional Ising model in a transverse field with periodic boundary condition as our working example. The associated Hamiltonian $H$ is given by
\begin{equation}
    H=-J\sum_i Z_i\otimes Z_{i+1}-B\sum_iX_i.
\end{equation}
$Z_i, X_i$ are the Pauli matrices, \rv{$J$ is the coupling strength between neighboring spins, and $B$ is the strength of the transverse magnetic field.} Clearly, $H$ itself represents a many-body system, but it can split into two simple terms
$h_1 = -J \sum_i Z_i Z_{i+1}, \quad h_2 = -B \sum_i X_i$, of which all the eigenstates and eigenvalues are explicitly known.  We choose this transverse Ising model for other two reasons.  First, though this is a many-body system, it is analytically solvable so that it can serve as a benchmarking model. Second,  if we add to this model another simple term such as  $\sum_i X_i\otimes X_{i+1}$, then we will have the so-called XX-ZZ model as a 2-local Hamiltonian in  the quantum-Merlin-Arthur (QMA)  complexity class \cite{PhysRevA.78.012352}.  Even  for the XX-ZZ model, it can also split into three or more simple terms of which all the eigenstates and eigenvalues are also known -- hence our algorithm equally applies.  

We now focus on the transverse Ising model. The associated eigenstates of $h_1$ are simply product states of the computational basis states $\ket{0}, \ket{1}$ and the eigenstates of $h_2$ are simply product states of the $X$-basis states $\ket{+}, \ket{-}$. 
  We set $J=1, B=0.5$ in dimensionless units.  The fidelity $F(\rho_1, \rho_2)=\mathrm{Tr}(\sqrt{\rho_1^{1/2}\rho_2\rho_1^{1/2}})$ is used as an indicator of the accuracy of our algorithm implemented.  Though the design of our algorithm always assumes a large scale quantum platform to be applied, we do not have direct access to quantum computers with many post-selection steps, and hence simulations below are based on classical computers as proof-in-principle demonstrations.   Therefore, the performance level indicated below will not take into account the potential errors incurred on an actual quantum hardware, e.g, when executing any imaginary time evolution operations. Rather, our simulations below, only capturing errors when discretizing the time variable, serve merely as a confirmation that the workflow of our algorithm works for a many-body quantum system using the transverse Ising model as an example.   Furthermore, in our numerical simulations below,  we do not really implement any post selection, instead we directly numerically apply operations related to low-lying states such as $e^{-\eta P_g}$ and hence there is no concern about the overall success rate in our numerical experiments. \rv{The simulation strictly adheres to the procedures outlined in Algorithms~\ref{alg:the_alg} and~\ref{alg:sb}. When computing the $i$-th excited state, the lifting probability $p_{\text{lift}}$ for all previously determined lower eigenstates is set uniformly to $\frac{1}{2i}$}

We begin with the case $L = 4$, using the proposed method to compute the lowest four eigenstates of the system. The corresponding eigenenergies are estimated by evaluating the expectation values of the Hamiltonian $H$ with respect to these states -- a simple step to implement physically because the energetics of $h_1$ and $h_2$ are very simple to evaluate, even in actual experiments. 

\rv{The simulation result of the case $L=4$ is shown in Fig.~\ref{4q}}
To compute the ground state $\rho_g$, we use Algorithm~\ref{alg:the_alg} with a total evolution time $T = 100$, step size $\eta = 0.05$, and a total number of ITE steps $N = 2000$. We observe that the energy decays exponentially over time, consistent with the pattern in standard imaginary time evolution algorithms. The final fidelity with respect to the exact ground state $\rho_0'$ is $99.9\%$, and the resulting (exact) eigenvalue is $-4.259497$ ($-4.271558$).

Next, we calculate the excited states, using the already computed low-lying eigenstates to construct the projectors.  We adopt Algorithm~\ref{alg:sb}:  For the first excited state $\rho_1$, we use $P_\text{lift} = 0.5$, with $T = 500$, $\eta = 0.05$, and $N = 10{,}000$. The fidelity with the exact first excited state is $97.1\%$. The resulting (exact) eigenvalue is $-4.230225$ ($-4.236068$).
For the second excited state $\rho_2$, we need to filter out both the ground state and the first excited state.  To that end, we introduce $P_\text{lift} = 0.25$ for both $\rho_g $ and $ \rho_1$ using $T = 60$, $\eta = 0.02$, and $N = 3000$. The fidelity of $\rho_2$ is $96.8\%$.
The resulting (exact) eigenvalue is $-1.314098$ ($-1.324307$).
For the third excited state $\rho_3$, we can now introduce $P_\text{lift} = 1/6$ for all lower eigenstates, with the same evolution parameters $T = 500$, $\eta = 0.02$, and $N = 25{,}000$. The resulting fidelity for $\rho_3$ is also $96.6\%$.
The resulting (exact) eigenvalue is $-1.053744$ ($-1.000000$).  One can certainly aim to yield higher fidelities by discretizing the time into more steps.

To check what happens if we scale up the system size, we further investigate a case with larger system size $L = 10$. Expectedly, the total number of samples required in the stochastic sampling will exponentially increase (hence longer evolution time).  However,  it is notable that our algorithm continues to yield multiple energy eigenstates reliably, thus signaling its robustness.  The two lowest eigenstates  are almost degenerate, i.e. $e_0=-10.6356$, $e_1=-10.6352$, so we check if our algorithm can yield the three lowest-energy states, including the third lowest eigenstate of eigenvalue  $e_2=-8.4486$. Our focus is on the gap between $e_1$ and $e_2$, with the results shown in Fig.~\ref{q10}.
In calculating the ground state, we initialize the system with the product state $\ket{+}^L$.  After executing our algorithm, the ground state energy is found to be  $-10.5528$, with the fidelity of the ground state found to be $98.6\%$.   Next,
we use the ground state of $h_1$, that is the product state $\ket{0}^L$ to initialize the system (the overlap between the previous state $\ket{+}^L$ and the exact first excited state  is almost 0).  The first excited state energy obtained from our algorithm is $-10.510592$, with the state fidelity $97.8\%$. The third eigenstate is obtained the same way as in the case of $L=4$, with a fidelity of $90.2\%$, and its corresponding energy is found  to be $-8.140736$. The calculated energy gap between $e_1$ and $e_2$ is $2.3699$, close to the exact value of $2.1866$. Notably, although the first two eigenstates are nearly degenerate, the error introduced by the discretization in our ITE algorithm is still small enough to distinguish between these two eigenvalues. 

\begin{figure}[htbp]
	\includegraphics[scale=0.07]{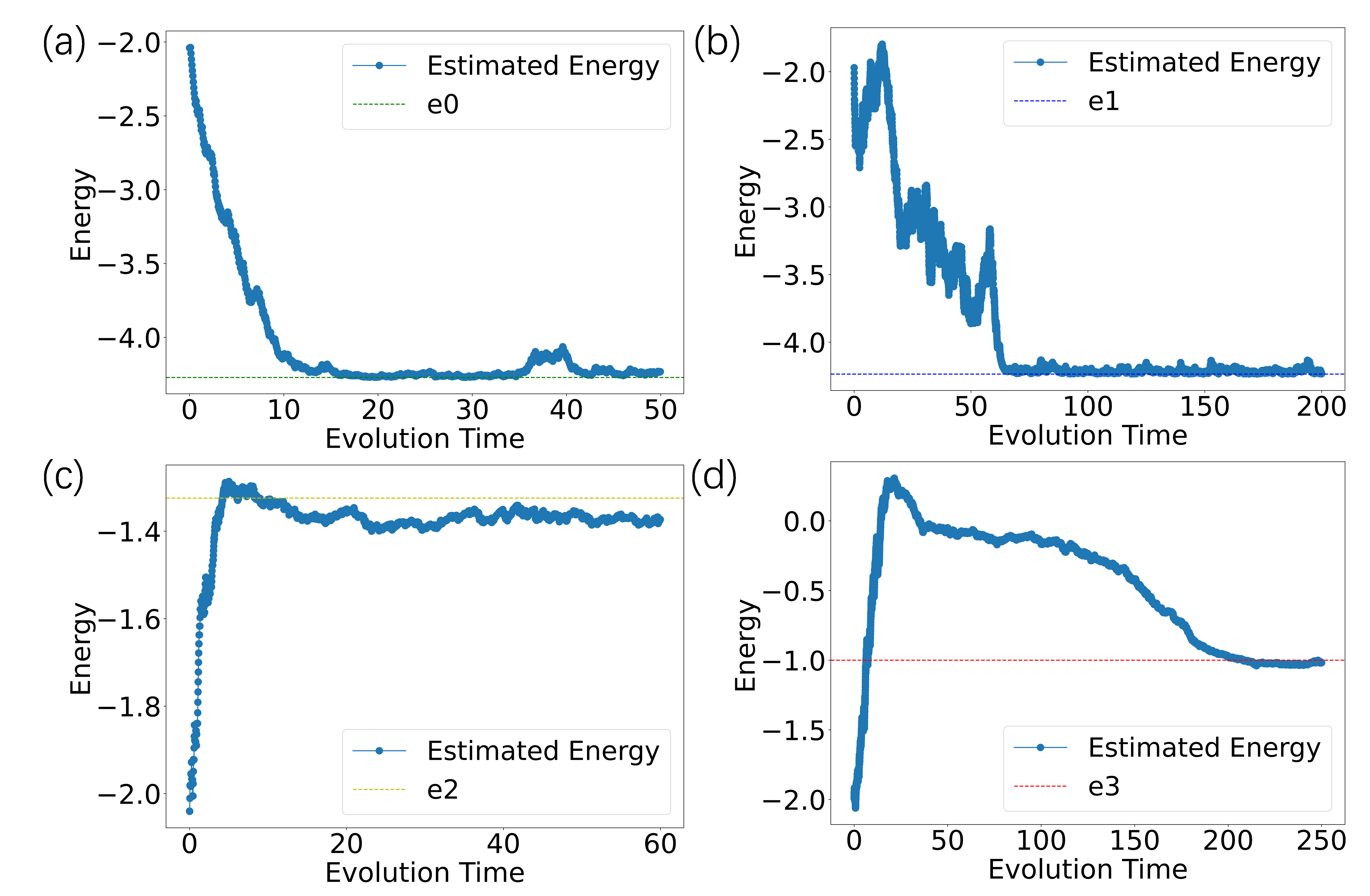}
	\caption{\label{4qpauli} Pauli operator sampling variant : Evolution of the first 4 lower eigenstates of the 4-qubit Ising medol with $J=1, B=0.5$. (a) The result for ground state. The fidelity is $99.5\%$ (b) The result for first excited state with the calculated ground state. The fidelity is $99.4\%$ (c) The result for second excited state with the calculated ground state and first excited state. The fidelity is $97.5\%$ (d) The result for third excited state with the calculated ground state and first and second excited state. The fidelity is $97.1\%$}
\end{figure}

\rv{Lastly, we show present the performance of the Pauli-operator sampling approach using the $L=4$ case, with results presented in Fig.~\ref{4qpauli}. In this implementation, each iteration involves sampling operators of the form $I-\eta Z_i\otimes Z_{i+1}$ or $I-\eta X_i$. For the ground state, we performed stochastic sampling ITE with $T=50$ steps and a step size of $\eta=0.05$, achieving a final fidelity of $99.5\%$ and an eigenvalue of $-4.231827$. To obtain the first excited state ($T=200, \eta=0.02$), the fidelity and eigenvalue reached $99.4\%$ and $-4.229948$, respectively. For the second excited state ($T=60, \eta=0.02$), the fidelity was $97.5\%$ with an energy of $-1.373191$. Finally, the third excited state was obtained with $T=250$ and $\eta=0.01$, resulting in a fidelity of $97.1\%$ and an eigenvalue of $-1.017456$.}

\section{Conclusion}

Solving the eigenvalues and eigenstates of a large Hermitian matrix (such as that of a many-body quantum system) has broad implications and applications in physics and beyond.  In this work,  we have shown that it is possible to construct a fully quantum hardware based algorithm to achieve this ambitious goal, going way beyond finding the ground state of a system Hamiltonian.   Our assumption about $H=h_1+h_2$ is not really restrictive, because many target Hamiltonians (including models in the QMA-complete complexity class \cite{PhysRevA.78.012352}) can indeed be written as a sum of simple components whose eigenvalues and eigenstates are explicitly known and easily measurable on quantum circuits. 

The proposed algorithm is made possible by a combination of two elements: stochastic sampling of simple projectors and a quantum circuit implementation of projectors of low-lying eigenstates already obtained.  Compared with the other ITE algorithms to date, we do not need exponentially expensive steps to design the unitary operations that can accommodate the needed imaginary time evolution. Compared with quantum annealing, we do not need to assume that along a chosen pathway, there is always the absence of dangerously small gap between the ground state and the excited state -- our algorithm works on individual Hamiltonians, for both the ground state and for the excited states.   Because the algorithm proposed here is a purely quantum algorithm capable of finding eigenstates of a large-scale quantum system, it is a concrete example of using quantum circuits to tackle classically intractable problems and will be  broadly applicable to a range of quantum tasks.   

 The overall success rate of our algorithm is a result of many post selections needed in the ITE and in the state-based simulation to implement projectors of already found states on the quantum circuit. How to improve the overall success rate associated with the ITE by examining  an optimal trade-off between accuracy and success rate for large scale problems remains to be investigated. Interestingly,  at least in terms of the success rate associated with the implementation of the projectors to filter out the already found ground state etc, we show that more qubits may not exponentially diminish the overall success rate because there are not many such steps required in the stochastic sampling based ITE operations. 

The work here will further motivate researches to look into potential challenges when our algorithm is implemented on an actual quantum hardware or when it is scaled up to solve real-life large-scale problems.   In particular, many physical systems do exhibit symmetries, and fully exploiting these symmetries to reduce computational cost represents one promising problem to be tackled with.   The actual performance of the proposed algorithm here executed on noisy quantum circuits, and how to mitigate the errors due to both quantum and classical noise must be further studied.  We also hope that the algorithm advocated in this work can motivate future algorithmic innovations targeted at truly large scale problems.  

\section{Acknowledgement}
 J.G. acknowledges support by the National Research Foundation, Singapore through the National Quantum Office, hosted in A*STAR, under its Centre for Quantum Technologies Funding Initiative (S24Q2d0009). J.G. also thanks Yijian Zou, Ching Hua Lee and Barry Sanders for helpful discussions.

\section{Data availability} 
The data are available from
the authors upon reasonable request.

\appendix

\section{Error analysis in stochastic sampling}

As the stochastic sampling approach is an approximation to ITE, it is necessary to analyze the error introduced in this step. Specifically, we consider three sources: the error in the imaginary time evolution, the error in the time-evolved state, and the deviation in the estimating the overall success rate as compared with that based on the exact ITE.

Here, $\Vert \cdot \Vert_1, \Vert \cdot \Vert_\infty$ are trace norm and spectral norm.

\begin{lemma}
(a general version of lemma.1 in \cite{ye2025quantum}) Let $A, B, C$ be matrices such that $\|A\|_1 = a$, and $\|B\|_\infty\leq b , \|C\|_\infty \leq c $, where a,b,c are non-negative constants. Then:
$$
\begin{aligned}
    \|BAB\|_1 &\leq b^2a, \\
    \|BAC + CAB\|_1 &\leq 2bca.
\end{aligned}
$$
\end{lemma}

\begin{proof}
Using Hölder's inequality for matrix norms with dual pairs $(p = 1, q = \infty)$, we have:

$$
\|BA\|_1 \leq \|B\|_\infty \|A\|_1, \quad \text{and} \quad \|AB\|_1 \leq \|A\|_1 \|B\|_\infty.
$$

Applying this iteratively:

$$
\|BAC\|_1 \leq \|B\|_\infty \|AC\|_1 \leq \|B\|_\infty \|C\|_\infty \|A\|_1 \leq a.
$$

Hence,

$$
\|BAB\|_1 \leq \|B\|_\infty \|AB\|_1 \leq \|B\|_\infty^2 \|A\|_1 \leq b^2a,
$$

and

$$
\|BAC + CAB\|_1 \leq \|BAC\|_1 + \|CAB\|_1 \leq 2bca.
$$
\end{proof}


We now analyze the error introduced by the approximate evolution operator. The total number of steps is $N$, with each step using a sampled projector, and the total evolution time is $T=N\eta$ with the step time $\eta \ll 1$. The expected evolution operator over $N$ steps is:
\begin{equation}
    E_{P_{1\dots N}}(I-\eta P_N)\dots(I-\eta P_1)=(I-\eta \rho)^N.
\end{equation}
The error compared to the exact ITE operator $e^{-N\eta \rho}$ is 
\begin{widetext}
\begin{align}
    e^{-N\eta \rho} =(I+(-\eta \rho)+\frac{1}{2}(-\eta\rho)^2 &+\dots)^N=(I-\eta \rho)^N+\frac{N}{2}(-\eta \rho)^2+O(N^2(\eta\rho)^3),\notag\\
    \Vert e^{-N\eta \rho}-(I-\eta \rho)^N \Vert_1&=O(N\eta^2 \mathrm{Tr}(\rho^2))\leq O(N\eta^2),
\end{align}
\end{widetext}
 where we used $\mathrm{Tr}(\rho^2)\leq 1$.

 Next, we consider the error in the evolved quantum state. We explore the error after one step time $\eta$ first. Let $\rho_1$ $\rho_2$ be two density matrices such that $\Vert \rho_1-\rho_2\Vert_1=\epsilon$. Then:
 \begin{widetext}
 \begin{align}
     &\Vert E_{P\sim \rho}(I-\eta P)\rho_1(I-\eta P)-e^{-\eta \rho}\rho_2e^{-\eta P}\Vert_1 \notag\\
     \leq \quad &  \Vert E_{P\sim \rho}(I-\eta P)(\rho_1-\rho_2)(I-\eta P)\Vert_1 +\Vert E_{p\sim \rho}(I-\eta P)\rho_2(I-\eta P)-e^{-\eta \rho}\rho_2e^{-\eta P}\Vert_1 \notag \\
     \leq \quad& E_{P\sim \rho}\Vert(I-\eta P)(\rho_1-\rho_2)(I-\eta P)\Vert_1 + \Vert E_{P\sim \rho}\eta^2P\rho_2P-\eta^2\rho \rho_2 \rho -R\rho_2R-(R\rho_2(I-\eta\rho)+(I-\eta\rho)\rho_2R) \Vert_1 \notag\\
     \leq \quad & \epsilon + 4\eta^2,
\end{align}
\end{widetext}
where $R=e^{-\eta\rho}-(I-\eta\rho)$ and $\Vert R \Vert_\infty \leq \frac{1}{2}\eta^2$

Now, suppose the initial state is $\rho_0$, the final (unnormalized) state after $N$ steps is is $\sigma_f(\rho_0)$ , while the exact ITE state is $\sigma_T(\rho)=e^{-N\eta\rho}\rho_0e^{-N\eta\rho}$. It follows that $\Vert \sigma_f(\rho_0)-\sigma_T(\rho)\Vert_1\leq 4N\eta^2$. And defining $\gamma =N\eta^2$, we wirte that:
\begin{equation}
    \sigma_f(\rho_0)-\sigma_T(\rho_0)=\gamma o
\end{equation}
with $\Vert o\Vert_1 =c\leq 4$

To estimate the difference between the normalized states, we consider:
\begin{align}
    &\Vert \frac{\sigma_T(\rho_0)+\gamma o}{\mathrm{Tr}(\sigma_T(\rho_0))+c\gamma}-\frac{\sigma_T(\rho_0)}{\mathrm{Tr}(\sigma_T(\rho_0))}\Vert_1 \notag\\
    = \quad& \gamma \Vert \frac{\mathrm{Tr}(\sigma_T(\rho_0))o-c\sigma_T(\rho_0)}{\mathrm{Tr}(\sigma_T(\rho_0))(\mathrm{Tr}(\sigma_T(\rho_0))+c\gamma)}\Vert_1\notag\\
   \leq \quad & \gamma \frac{2|c|}{|(\mathrm{Tr}(\sigma_T(\rho_0))+c\gamma)|}\notag\\
  = \quad & O(\gamma),
\end{align}
where we have assumed $\gamma << e^{-2T \lambda_{\text{min}}}$, with $\lambda_{\text{min}}$ being the minimal eigenvalue of $\rho$.
Finally, we can conclude that the success probability $\mathrm{Tr}(\sigma_f(\rho_0))$ can be estimated as the following: $|\mathrm{Tr}(\sigma_f(\rho_0)-\sigma_T(\rho_0))|\leq O(\gamma)$

\ 
\vspace{1cm}  \section{Errors in finding eigenstates
above the ground state }

In the process of finding eigenstates
above the ground state, the state based simulation is applied to deal with $P_g$ and the the obtained ground state should be slightly different from the exact one. Therefore, we discuss here the error from the state based simulation and the imperfect ground state used in our algorithm. 

For the state based simulation, we assume that the error of two initial state is $\Vert \rho_1-\rho_2\Vert_1=\epsilon$.
Then the error after one step (we ignore the prefactor $\frac{1}{2}$ for simplicity) is
\begin{widetext}
\begin{align}
    &\Vert (\rho_1 - \eta \{ \rho_1, P_g \}+\eta^2P_g)-(I-\eta P_g) \rho_2 (I-\eta P_g)\Vert_1 \notag\\
    \leq \quad & \Vert (I-\eta P_g)(\rho_1- \rho_2) (I-\eta P_g)\Vert_1+ \Vert (\rho_1 - \eta \{ \rho_1, P_g \}+\eta^2P_g)-(I-\eta P_g) \rho_1 (I-\eta P_g) \Vert_1 \notag\\
    \leq \quad&  \epsilon+ \Vert \eta^2 P_g (I-\rho_1) P_g\} \Vert_1  \notag\\
    = \quad & \epsilon+  \eta^2 \text{Tr}((I-\rho_1) P_g)  \notag\\
    \leq \quad& \epsilon +\eta^2,
\end{align}
\end{widetext}
which gives us the desired unnormalized transformed state. The success rate is $\frac{\mathrm{Tr}((I-\eta P_g) \rho_0 (I-\eta P_g)}{2}+c_1\eta^2$ with $|c_1|\leq 1/2$.

At the end of this section, we analyze the error introduced when the ground state used in the lifting procedure is not ideal. Specifically, we consider the scenario where the calculated ground state is used to construct a new Hamiltonian, with the goal of treating the first excited state as the new effective ground state. Since we are primarily interested in the impact of ground state imperfections on the excited state, the problem can be simplified to a two-level (two-band) model. Without loss of generality, we assume that $\ket{0}$ is the true ground state and $\ket{1}$ is the first excited state, separated by an energy gap $b > 0$. The corresponding two-band Hamiltonian can be written as

\begin{equation}
    H = 
\begin{bmatrix}
0 & 0 \\
0 & b  
\end{bmatrix}.
\end{equation}

In practice, the obtained ground state is only approximate and typically has some overlap with the exact excited state. Let the approximate (perturbed) ground state be $\ket{\psi_g}=\ket{0} + \delta\ket{1}$, where $\delta$ quantifies the error, i.e., the overlap with $\ket{1}$. Suppose we lift the ground-state energy by an amount $a>0$. The resulting effective Hamiltonian for excited-state extraction is

\begin{align}
H' &= b\ket{0}\bra{0}+a\ket{\psi_g}\bra{\psi_g} \notag\\
&=
\begin{bmatrix}
a & a\delta \\
a\delta & b + a\delta^2 
\end{bmatrix} =
\begin{bmatrix}
a & 0 \\
0 & b + a\delta^2 
\end{bmatrix} +
a\delta \begin{bmatrix}
0 & 1 \\
1 & 0 
\end{bmatrix}.
\end{align}

Assuming the energy lifting is sufficiently large and the ground-state error is small (specifically, $a \ge 2b$ and $\delta \ll 1$), the off-diagonal term can be regarded as a small perturbation, allowing us to apply perturbation theory. The calculated (perturbed) excited state becomes:

\begin{equation}
    \ket{1'} = \ket{1} - \frac{a\delta}{a - b - a\delta^2} \ket{0}.
\end{equation}

The corresponding error $\delta'$ in the excited state is given by:

\begin{equation}
    \delta<\delta' = \left| \frac{a\delta}{a - b - a\delta^2} \right| \leq \frac{2\delta}{1-2\delta^2}\leq 2\delta+6\delta^3 \approx 2\delta.
\end{equation}
We use the relation $(1-2\delta^2)(1+3\delta^2)\ge 1$ in the third inequality as $\delta \ll 1$.

This result implies that the error introduced by lifting the imperfect ground state propagates to the excited state with an approximate factor 2. As we climb higher in the energy spectrum, this error can accumulate exponentially ($\approx 2^l \delta$ for the $l$-th excited state). Therefore, to accurately extract higher excited states, the ITE evolution must be performed for a sufficiently long time to ensure that the ground state error is reduced to the order of $2^{-l}$, where $l$ is the desired level. 
Therefore, in practice, with our algorithm we only advocate to calculate a limited number of low-lying energy levels to maintain an acceptable accuracy.

\section{resource estimation}

\rv{In this section, we provide a resource analysis for practical implementation, specifically focusing on the number of time steps $N$, the step time $\eta$, the total success probability, and the two-qubit gate count required per step. For simplicity, we define the ground state $\ket{g}$ as the target state and initialize the system in state $\rho_0$ with an initial fidelity $f_0 = \text{Tr}(\sqrt{\sqrt{\rho_0} P_g \sqrt{\rho_0}})$. Suppose in our algorithm the obtain final state is $\rho_T$ after an evolution time $T=N\eta$. The fidelity of $\rho_T$ relative to the ideal target is $f_T$ and we assume the allowed error in this fidelity is bounded by $\Delta f$. In this resource estimation, we first analyze the performance of ideal ITE. By incorporating the previously established differences between the ideal evolution and our proposed method, we then derive the total resource requirements.}

\rv{We define the mixed-state density matrix associated with the Hamiltonian $H$ as $\rho_H$, characterized by a ground-state eigenvalue $\lambda_0$ and a spectral gap $\Delta\lambda$ between the ground and the first excited state. The ideal ITE based on $\rho_H$ yields the final state $\rho'_T = e^{-T\rho_H} \rho_0 e^{-T\rho_H} / \mathcal{Z}$, where $\mathcal{Z} = \text{Tr}(e^{-T\rho_H} \rho_0 e^{-T\rho_H})$ is the normalization factor. We assume that the fidelity of $\rho'_T$ is $f'_T$. Without loss of generality, this evolution can be expressed as:
\begin{equation}
\begin{split}
    \rho_0 &= (f_0\ket{g}+\sqrt{1-f_0^2}e^{i\phi_0}\ket{e})(f_0\bra{g}+\sqrt{1-f_0^2}e^{i\phi_0}\bra{e}) \\
    \rho'_T &= (f'_T\ket{g}+\sqrt{1-{f'}_T^2}e^{i\phi_0}\ket{e})(f'_T\bra{g}+\sqrt{1-{f'}_T^2}e^{i\phi_0}\bra{e})
\end{split}
\end{equation}
where $\ket{e}$ denotes the subspace spanned by all excited states and we take the same phase $\phi_0$ because the ideal ITE implementation will not modify the relative phase between $|g\rangle$ and $|e\rangle$.  To ensure that the final state $\rho'_T$ achieves a fidelity $f'_T>f_T$, the evolution time $T$ must satisfy:
\begin{equation}
\begin{split}
    &\frac{f'_T}{\sqrt{1-{f'}_T^2} }=e^{T\Delta\lambda}\frac{f_0}{\sqrt{1-f_0^2} }\ge \frac{f_T}{\sqrt{1-f_T^2} }\\
    &T\ge\frac{1}{\Delta\lambda}\ln{\left(\frac{f_T\sqrt{1-f_0^2}}{f_0\sqrt{1-f_T^2}}\right)}.
\end{split}
\end{equation}}

\rv{In our actual implementation of the stochastic sampling-based ITE where we have discretized time steps, the final state $\rho_T$ from our algorithm differs from $\rho'_T$, that is $\Vert \rho_T-\rho'_T\Vert_1 =O(\gamma)$, where $\gamma=N\eta^2$. The difference in the fidelity squared between $\rho_T$ and $\rho_T'$ is given by
\begin{align}
|\text{Tr}(\rho_TP_g)-\text{Tr}(\rho'_TP_g) | =&\ \frac{1}{2}|\text{Tr}(\Delta_{\rho_T}P_g+P_g\Delta_{\rho_T}) |\notag\\
    \leq & \ \frac{1}{2}\Vert\Delta_{\rho_T}P_g+P_g\Delta_{\rho_T}\Vert_1\notag\\
    \leq&\ \Vert\Delta_{\rho_T}P_g\Vert_1\notag\\
    \leq&\ \Vert\Delta_{\rho_T}\Vert_1\notag=\ O(\gamma),
\end{align}
where $\Delta_{\rho_T}=\rho_T-\rho'_T$. The first inequality follows from the property that the trace norm of a Hermitian matrix is bounded below by the absolute value of its normalized trace. The third inequality is derived using Hölder's inequality, noting that the trace norm of the ground-state projector is unity, i.e., $\Vert P_g\Vert_1=1$. 
Hence, if we choose $\gamma$ based on a given $\Delta f$ through $\gamma\approx \Delta f$, then all the above inequalities holds. }

\rv{With these considerations and choices, it is then clear that to achieve the fidelity $f_T\pm \Delta f$ with our algorithm, the key implementation parameters are:
\begin{equation}
    \begin{split}
        \eta &=\frac{\gamma}{T}=\frac{\Delta\lambda\Delta f}{\ln{\left(\frac{f_T\sqrt{1-f_0^2}}{f_0\sqrt{1-f_T^2}}\right)}};\\
        N &= \frac{T}{\eta}=\frac{1}{\Delta f\Delta^2\lambda}\ln^2{\left(\frac{f_T\sqrt{1-f_0^2}}{f_0\sqrt{1-f_T^2}}\right)}.
    \end{split}
\end{equation}}

\rv{The total success probability is 
\begin{align}
    &\text{Tr}(e^{-T\rho_H}\rho_0 e^{-T\rho_H})+O(\gamma)\notag\\
    \le &~ e^{-2\lambda_0T}+O(\Delta f)\notag\\
    =&~ \left(\frac{f_T\sqrt{1-f_0^2}}{f_0\sqrt{1-f_T^2}}\right)^{-2\frac{\lambda_0}{\Delta\lambda}}+O(\Delta f).
\end{align}
At each step, there are $O(n)$ two-qubit gates needed. }

%

\end{document}